\documentclass[
final
]{dmtcs-episciences}

\usepackage[utf8]{inputenc}
\usepackage{subfigure}

\usepackage{graphicx}
\usepackage[dvipsnames]{xcolor}
\usepackage{amsthm}
\usepackage{amsmath}
\usepackage{amssymb}
\usepackage[margin=1.25in]{geometry}
\usepackage[sf,bf,small,raggedright,compact]{titlesec}
\usepackage{hyperref}
\definecolor{linkblue}{named}{black}
\hypersetup{colorlinks=true, linkcolor=linkblue,  anchorcolor=linkblue,
        citecolor=linkblue, filecolor=linkblue, menucolor=linkblue,
        urlcolor=linkblue}
\usepackage[capitalize]{cleveref}
\usepackage{paralist}
\usepackage[longnamesfirst,numbers,sort&compress]{natbib}

\usepackage{listings,newtxtt}
\newtheorem{thm}{Theorem}
\newtheorem{obs}{Observation}
\newtheorem{lem}{Lemma}
\newtheorem{cor}{Corollary}
\newtheorem{conj}{Conjecture}
\newtheorem{prb}{Problem}
\newtheorem{rem}{Remark}
\newtheorem{prop}{Proposition}
\newtheorem*{clm}{Claim}

\newcommand{\defin}[1]{\emph{\textcolor{black}{#1}}}

\author[Biniaz et. al]{
    \centering
   Ahmad Biniaz\affiliationmark{1}
  \and Prosenjit Bose\affiliationmark{2}
  \and Jean-Lou De Carufel\affiliationmark{3}
  \and Anil Maheshwari\affiliationmark{2}
  \and Babak Miraftab\affiliationmark{2}
  \and Saeed Odak\affiliationmark{3}
  \and Michiel Smid\affiliationmark{2}
  \and Shakhar Smorodinsky\affiliationmark{4}
  \and Yelena Yuditsky\affiliationmark{5}
}
\title[On Separating Path and Tree Systems in Graphs]{On Separating Path and Tree Systems in Graphs\thanks{Ahmad Biniaz, Prosenjit Bose, Jean-Lou De Carufel, Anil Maheshwari,
Babak Miraftab, Saeed Odak, and Michiel Smid were in part supported by NSERC. Shakhar Smorodinsky was partially supported by the Israel Science Foundation (grant no.~1065/20) and by the United States -- Israel Binational Science Foundation (NSF-BSF grant no.~2022792). Yelena Yuditsky  was supported by the Belgian National Fund for Scientific Research (FNRS).}}
\affiliation{
  School of Computer Science, University of Windsor, Windsor, Canada\\
  School of Computer Science, Carleton University, Ottawa, Canada\\
  School of Electrical Engineering and Computer Science, University of Ottawa, Ottawa, Canada\\
  Department of Computer-Science, Ben-Gurion University of the Negev, Beersheba, Israel\\
  D{\'e}partement d’Informatique, Universit{\'e} libre de Bruxelles, Brussels, Belgium
}
\begin{document}
% This is only used if you are compiling for a volume before vol 25
% \publicationdetails{VOL}{2015}{ISS}{NUM}{SUBM}
% This is the new form of collecting the data, starting with vol 25
\publicationdata{vol. 27:2}{2025}{12}{10.46298/dmtcs.12743}{2023-12-25; 2023-12-25; 2025-02-03}{2025-04-24}

\maketitle

\begin{abstract}
We explore the concept of separating systems of vertex sets of graphs. A separating system of a set $X$ is a collection of subsets of $X$ such that for any pair of distinct elements in $X$, there exists a set in the separating system that contains exactly one of the two elements. A separating system of the vertex set of a graph $G$ is called a vertex-separating path (tree) system of $G$ if the elements of the separating system are paths (trees) in the graph $G$. In this paper, we focus on the size of the smallest vertex-separating path (tree) system for different types of graphs, including trees, grids, and maximal outerplanar graphs.
\end{abstract}

\section{Introduction}

Given a set $X$, a collection $\mathcal{F}$ of subsets of $X$ is called a \defin{weakly separating system} or, simply, a \defin{separating system} of $X$ if for every pair of distinct elements, $a, b \in X$, there exists a set $F \in \mathcal{F}$ such that $F$ separates $a$ and $b$, that is, $F$ contains exactly one of $a$ and $b$. We refer to the smallest size of a separating system as the \defin{separation number} of $X$. 

R\'{e}nyi \cite{renyi1961random} initiated this notion of separation in 1961. In fact, the separation number of a set of size $n$ is $\lceil \log{n} \rceil$. %\footnote{In this writing, all occurrences of $\log$ represent logarithms with a base of 2.}
By restricting the set $X$ and enforcing some conditions on the elements of $\mathcal{F}$, several interesting variants of separating set system problems have been studied in the literature \cite{DBLP:journals/corr/abs-2301-08707, DBLP:journals/dam/BaloghCMP16, wickes2022separating, DBLP:journals/combinatorics/Rosendahl03, DBLP:journals/dm/Cai84, DBLP:journals/ejc/BollobasS07, DBLP:journals/dm/Wegener79, DBLP:journals/dam/WienerHT20, DBLP:journals/jct/KundgenMT01}. For example, when the set $X$ represents the edge set of a given $n$-vertex connected graph $G$, an \defin{edge-separating path system} of $G$ is defined as a collection of paths in $G$ that separates the edges of $G$, that is, the elements of the separating system of $X$ are paths in $G$. Falgas-Ravry et al. \cite{falgas2013separating} conjectured that the smallest size of an edge-separating path system of $G$ is independent of the size of the edge set and it is linear in the size of the vertex set, i.e. $O(n)$; this has been proven recently in \cite{DBLP:journals/corr/abs-2301-08707}. Also, geometric versions of this problem, when the set $X$ is an arbitrary point set in the plane and the separating sets are geometric objects, like circles and convex sets, have been studied in \cite{DBLP:journals/comgeo/GerbnerT13} and \cite{DBLP:journals/dm/GledelP19}.

In this paper, we focus on separating the vertex set of a graph with paths and trees. Let $G$ be a graph and $X \subseteq V(G)$ be a set of vertices of $G$. A collection of distinct paths $\mathcal{S} := \{\Pi_1, \Pi_2, \Pi_3, \ldots, \Pi_k\}$ in $G$ is called a \defin{separating path system} of $X$ if for every pair of distinct vertices, $u, v \in X$, there exists an $\ell \in \{1, 2, 3, \ldots, k\}$ such that $\Pi_\ell$ contains exactly one of $u$ and $v$. We define a \defin{vertex-separating path system} of the graph $G$ to be a separating path system of $V(G)$. We are interested in the size of the smallest vertex-separating path system of the graph $G$, and we denote this number by $f(G)$. Analogously, we define a \defin{vertex-separating tree system} of $G$, when the elements of the separating set $\mathcal{S}$ are tree subgraphs of $G$, and we denote the size of the smallest vertex-separating tree system of $G$ by $f_t(G)$.

To the best of our knowledge, there exist only a few results in the literature introducing and studying vertex-separating path systems. Foucaud and Kov{\v{s}}e \cite{DBLP:journals/jda/FoucaudK13} studied this parameter in the context of identifying codes, and they provide optimal vertex-separating path systems for path and cycle graphs. They also present the first upper and lower bounds for $f(T)$ when $T$ is a tree. Recently, Arrepol et al. \cite{arrepol2023separating}, among other variants, studied vertex-separating path systems of random graphs, and they also improved the upper and lower bounds of $f(T)$ when $T$ is a tree. In \cref{trees}, we briefly review the known bounds for trees. We then present a tight lower bound for the size of the smallest vertex-separating path system. In fact, we show that $f(T) \ge \frac{n}{4}$ for every $n$-vertex tree.

\begin{table}[t]
  \renewcommand{\arraystretch}{1.5}
  \centering{
  \resizebox{\textwidth}{!}{
    \begin{tabular}{p{6cm}p{0.45cm}ccl} \hline
    Graph class & & Lower Bound & Upper Bound & Ref. \\ \hline
    $n$-vertex complete graph, $K_n$ & & $\lceil \log{n} \rceil$ & $\lceil \log{n} \rceil$ & \cite{renyi1961random} \\
    $d$-dimensional hypercube, $Q_d$ & & $d$ & $d$ & \cite{DBLP:journals/jda/FoucaudK13} \\
    $n$-vertex path and cycle, $P_n$, $C_n$ & & $\lceil \frac{n}{2} \rceil$ & $\lceil \frac{n}{2} \rceil$ & \cite{DBLP:journals/jda/FoucaudK13} \\
    Complete bipartite graph, $K_{m,n}$ & & $\Omega(\frac{n}{m}\cdot \frac{\log{n}}{\log{(1+n/m)}})$ & $O(\frac{n}{m}\cdot \frac{\log{n}}{\log{(1+n/m)}})$ & \cite{katona1966separating}, \cref{kmn} \\
    $m \times n$- grid graph, $G_{m,n}$ $(m, n \ge 2)$ & & $\lceil \log{m} + \log{n} \rceil$ & $2 \lceil \log{m} \rceil + 2 \lceil \log{n} \rceil$ & \cref{grid} \\
    Erd\H{o}s–R\'{e}nyi random graph $G(n, p)$, $p \ge (2\ln{n}+\omega(\ln{\ln{n}}))/n$ & & $\lceil \log{n} \rceil$ & w.h.p. $\lceil \log{n} \rceil+1$ & \cite{arrepol2023separating}\\
    Erd\H{o}s–R\'{e}nyi random graph $G(n, p)$, $p \le (\ln{n}-\omega(\ln{\ln{n}}))/n$ & & w.h.p. $\omega(\log{n})$ & $O(n)$ & \cite{arrepol2023separating}\\
    $n$-vertex tree & & $\frac{n}{4}$ & $\frac{2n}{3} + O(1)$ & \cite{DBLP:journals/jda/FoucaudK13}, \cref{trees_label} \\
    $n$-vertex maximal outerplanar graph & & $\Omega(\log{n})$ & $\frac{n}{4} + O(1)$ & \cref{outerthm} \\
    $n$-vertex $K_{2,t}$-minor-free graph with a vertex of degree $\Omega(n^\delta)$ & &  $\Omega(n^{\delta/{(t+1)}})$ & $\frac{2n}{3} + O(1)$ & \cref{poly-lower} \\
    \hline
    \end{tabular}
  } % centering
  }
  \caption{Summary of previous and new results on the value of $f(G)$ (w.h.p. stands for with high probability).}
  \label{summary-table}
\end{table}

In \cref{grid_sec}, we focus on separating the vertices of an $n$ by $n$ grid using $O(\log{n})$ paths, where $n \ge 2$. In a related study, Honkala et al. \cite{DBLP:journals/dam/HonkalaKL03} use cycles to separate the vertices of a torus. Notably, Rosendahl \cite{DBLP:journals/combinatorics/Rosendahl03} also studies the same parameter as discussed in \cite{DBLP:journals/dam/HonkalaKL03} but for higher dimensions. Additionally, Rosendahl \cite{DBLP:journals/combinatorics/Rosendahl03} analyze the precise values of $f(K_{n,n})$ and $f(K_{3,n})$, where $n \in \mathbb{N}$. In \cref{complete}, by utilizing an old result by Katona \cite{katona1966separating}, we present a tight asymptotic bound for the value of $f(K_{m,n})$, where $m$ and $n$ are arbitrary positive integers.

In addition to the aforementioned results, this paper also includes a tight upper bound for the value of $f(G)$ when $G$ is in the class of maximal outerplanar graphs. Moreover, we show that a $K_{2,t}$-minor-free graph with a high-degree vertex requires a polynomial-sized vertex-separating path system for any constant $t$. Next, we focus on the size of optimal vertex-separating tree systems and we prove that every $n$-vertex graph with radius $r$ has a vertex-separating tree system of size at most $r + 2\log{n} + 1$.

The rest of this paper is organized as follows: In section 2, we start with preliminaries and some simple observations related to vertex-separating path systems. In \cref{complete}, \cref{trees}, \cref{grid_sec}, and \cref{outer}, we will discuss vertex-separating path systems of complete bipartite graphs, trees, grids, and maximal outerplanar graphs, respectively (see \cref{summary-table}). \cref{poly-lowerbound} will establish a sufficient condition that guarantees a polynomial-sized lower bound for the size of vertex-separating path systems in certain classes of graphs. \cref{tree_sep} studies vertex-separating tree systems and compares this variant with vertex-separating path systems. Finally, we discuss some open problems in \cref{final}.

\section{Preliminaries} \label{preliminaries}

For integers $0 \le a \le b$, put $[a] := \{x | x \in \mathbb{Z}$ and $1 \le x \le a \}$ and $[a, b] := \{x | x \in \mathbb{Z}$ and $a \le x \le b\}$. Throughout this paper, we use standard graph theoretic terminology as used in the textbook by Diestel \cite{diestel:graph}. All graphs discussed here are connected, simple, finite, and have at least 4 vertices. We denote the vertex set and edge set of a graph $G$ by $V(G)$ and $E(G)$, respectively. We say that a subgraph $G'$ of a graph $G$ \defin{spans} a set $S \subseteq V(G)$ if $S \subseteq V(G')$. 

A \defin{path} in $G$ is a sequence of distinct vertices $v_0,v_1,\ldots,v_r$ with the property that $\{v_{i-1},v_i\}\in E(G)$, for each $i \in [r]$.  The \defin{endpoints} of such a path are the vertices $v_0$ and $v_r$. The \defin{length} of a path is the number of edges in the path. A path of length zero is called a \defin{trivial path}. If $v$ is the only vertex in a trivial path $\Pi$, we say $v$ \defin{creates} $\Pi$. A vertex $v \in V(G)$ is called a \defin{center} of the graph $G$ if the largest distance of $v$ to other vertices in $V(G)\setminus\{v\}$ is minimal (a center vertex may not be unique).

An \defin{$m\times n$ grid} $G_{m, n}$ is a graph with vertex set $V(G_{m, n}):=\{0,1,2,\ldots, m-1\}\times\{0,1,2, \ldots, n-1\}$ and edge set $E(G_{m,n}) = \{(i, j)(i', j') | 0 \le i, i' < m, 0 \le j, j' < n$ and $|i - i'| + |j - j'| = 1\}$. For a vertex $(i, j) \in V(G_{m,n})$, we define the projection functions as $\pi_x(i, j) = i$ and $\pi_y(i, j) = j$. Each \defin{column} of $G_{m,n}$ consists of vertices with the same $\pi_y$ value, that is the vertex set $\{0,1,2\ldots,m-1\}\times\{j\}$ for some fixed $j \in \{0,1,2, \ldots, n-1\}$. Similarly, each \defin{row} of $G_{m,n}$ consists of vertices with the same $\pi_x$ value, that is the vertex set $\{i\}\times \{0,1,2,\ldots,n-1\}$ for some fixed $i \in \{0,1,2, \ldots, m-1\}$. A set $C$ of columns (rows) is consecutive if $G_{m,n}[\cup C]$ is connected.% The $m\times n$ grid is called trivial if either $m = 1$ or $n = 1$.

Let $G$ be a graph and let $X$ be a non-empty subset of $V(G)$. A \defin{labeling} of $X$ is a function $\psi : S \rightarrow [0, 2^{\lceil \log{|X|} \rceil}-1]$. This labeling is called \defin{nice} if for each $1 \le i \le \lceil \log{|X|} \rceil$, the graph induced by $(V(G)\setminus X)\cup X_i$ contains a path that spans the set $X_i := \{v \in S$ $|$ the $i$-th bit in the binary representation of $\psi(v)$ is $1\}$.

\begin{thm}[\cite{renyi1961random}] \label{kn}
    Let $G$ be a graph and let $X$ be a non-empty subset of $V(G)$. Then $X$ has a separating path system of size $\lceil \log{|X|} \rceil$ if and only if $X$ has a nice labeling.
\end{thm}

Let $n$ and $k$ be positive integers. A separating path system of size $k$ can separate at most $2^k$ vertices from each other, therefore for any graph $G$ with $n$ vertices, $f(G) \ge \lceil \log{n} \rceil$ \cite{renyi1961random}. Since every induced subgraph of $K_n$ has a spanning path, by taking $X = V(G)$ in \cref{kn}, we have $f(K_n) = \lceil \log{n} \rceil$ (Note that any labeling of $V(G)$ is nice). If we are aiming to cover $V(G)$ using the same set of paths in the separating system, we need at least $\lceil \log{(n+1)} \rceil$ paths.\footnote{This is equivalent to not using the label $0$ for any vertex in $V(G)$ (cf. \cite[Proposition 2]{DBLP:journals/jda/FoucaudK13}).} By \cref{kn}, we have that $f(Q_k) = k$, where $Q_k$ is the $k$-dimensional hypercube \cite{DBLP:journals/jda/FoucaudK13}. As we will be referring to them in the sequel, we state the value of the parameter $f$ for paths and cycles.

\begin{obs} (\cite[Theorem 15]{DBLP:journals/jda/FoucaudK13}) \label{cycle}
For every integer $n \ge 3$, $f(P_n) = f(C_n) = \lceil \frac{n}{2} \rceil$, where $P_n$ and $C_n$ denote an $n$-vertex path and cycle, respectively.
\end{obs}

\section{Complete Bipartite Graphs} \label{complete}

To initiate our study, we consider separating path systems of the vertices of complete bipartite graphs $K_{m,n}$. It is worth mentioning that vertex-separating cycle systems of $K_{m,n}$ have been studied in \cite{DBLP:journals/combinatorics/Rosendahl03}. Let $m$ and $n$ be two positive integers with $m \le n$. Note that the length of the longest path in $K_{m,n}$ is $2m$. Katona \cite{katona1966separating} provided bounds for the size of the smallest separating set system of $[n]$ where each set in the separating system has size at most $1 \le k \le n$. In fact, he showed that if $\tau(n, k)$ is the size of the smallest such separating set system, then $\frac{n}{k} \cdot \frac{\log{n}}{\log{(en/k)}} \le \tau(n, k) \le \frac{n}{k} \cdot \frac{\log{2n}}{\log{(n/k)}}$. From this result, we note that $f(K_{m,n}) = \Theta(\frac{n}{m}\cdot \frac{\log{n}}{\log{(1+n/m)}})$. Since our construction is somewhat simpler, we provide a different proof for the upper bound of $f(K_{m,n})$.

% \begin{prop} \label{kmn}
%     Let $m$ and $n$ be two positive integers with $m \le n$. Then $f(K_{m,n}) = O(\frac{n}{m}\cdot \frac{\log{n}}{\log{(1+n/m)}})$.
% \end{prop}

\begin{prop} \label{kmn}
    $f(K_{m,n}) = O(\frac{n}{m}\cdot \frac{\log{n}}{\log{(1+n/m)}})$, for integers $n \ge m > 0$.
\end{prop}

\begin{proof}
We build a vertex-separating path system of the given size. Let $L$ refer to the $m$ vertices on the smaller part of $K_{m,n}$ and $R$ refer to the other part of size $n$. For subsets $X \subseteq L$ and $Y \subseteq R$, we denote the induced subgraph of $K_{m,n}$ on $X \cup Y$ by $K(X, Y) := K_{m,n}[X \cup Y]$. In order to separate the vertices of $L$ from $R$, consider a subset $R' \subseteq R$ of size $m$. Partition $L$ into two almost equal size sets, $L_1$ and $L_2$. Similarly, partition $R'$ into $R'_1$ and $R'_2$. By adding four paths covering the vertex sets of $K(L_1, R'_1)$, $K(L_1, R'_2)$, $K(L_2, R'_1)$, and $K(L_2, R'_2)$, we separate the vertices of $L$ from the vertices of $R$.

Next, if $m = n$, then \cref{kn} implies that $f(K_{m,m}) = \Theta(\log{m})$. More precisely, we separate the vertices of $L$ from each other using a nice labeling of the vertices $L$. That is we assign labels $0, 1, \dots, m-1$ to the vertices of $L$ in arbitrary order. For each $1 \le i \le \lceil \log{m} \rceil$, using vertices in $R$, we create a path that spans only those vertices in $L$ whose $i$-th bit in the binary representation of their labels is $1$. Similarly, we separate the vertices of $R$ from each other.

Now assume $n > m$. First, similar to the case $K_{m,m}$, using $O(\log{m})$ paths, we separate the vertices in $L$ from each other. It remains to separate the vertices within $R$ from each other. To achieve this, define $K := \lceil \frac{n}{m} \rceil$ and consider an auxiliary tree $\mathcal{T}$ having the elements of $R$ at its leaves, in which each internal node has $K$ children. We refer to the vertices of $\mathcal{T}$ as nodes. Note that the height of $\mathcal{T}$ is $O(\frac{\log{n}}{\log{K}})$. In order to simplify the explanation, we assume that $n$ is a power of $K$ (otherwise, each internal level of $\mathcal{T}$ will contain at most one node with less than $K$ children). For each node $u$ in this tree, let $S_u$ denote the set of elements of $R$ that are stored in the subtree rooted at $u$.

For each internal level $\ell$ in the tree, and for each $i=1,2,3,\ldots,K$, let $R(\ell,i)$ be the union of all sets $S_u$, where $u$ ranges over the $i$-th child of all nodes at level $\ell$. Note that the size of $R(\ell,i)$ is at most $m$. We add a path in $K(L,R(\ell,i))$ that spans $R(\ell,i)$ to the separating path system.  

Now, we show that the set of selected paths in the previous paragraph separates the vertices of $R$ from each other. Let $r$ and $r'$ be two distinct vertices in $R$. Let $v$ and $v'$ be the leaves of $\mathcal{T}$ that store $r$ and $r'$, respectively. Let $u$ be the lowest common ancestor of $v$ and $v'$ in $\mathcal{T}$, and let $\ell$ be the level of $u$. Let $i$ be such that $r$ is in the subtree rooted at the $i$-th child of $u$, and define $i'$ similarly with respect to $r'$. Since $i \neq i'$, the vertices $r$ and $r'$ are separated by the spanning path in $K(L,R(\ell,i))$.

% It will be easy to determine the number of paths added.
In this construction, $\mathcal{T}$ has $O(\frac{\log{n}}{\log{K}})$ levels, and in each level we select $O(K)$ paths. Therefore, the total number of paths used to separate the vertices of $K_{m,n}$ is $O(\log{m} + K \cdot \frac{\log{n}}{\log{K}}) = O(\log{m} + \frac{n}{m}\cdot \frac{\log{n}}{\log{(n/m)}}) = O(\frac{n}{m}\cdot \frac{\log{n}}{\log{(1+n/m)}})$, since $K = \lceil \frac{n}{m} \rceil$.
\end{proof}

As a corollary of \cref{kmn}, one can show that the value of the parameter $f(G)$ for $n$-vertex graphs $G$, asymptotically, can be as large as any sub-linear polynomial on $n$.

\begin{cor}
    Let $0 < \delta \le 1$ be a real number. For any positive integers $m$ and $n$, with $m = n^{1-\delta}$,
    $f(K_{m, n}) = {\Theta}(n^{\delta})$.
\end{cor}

\section{Trees} \label{trees}

As mentioned in the introduction, the smallest size of vertex-separating path systems of trees has been studied by Foucaud and Kov{\v{s}}e \cite{DBLP:journals/jda/FoucaudK13} and Arrepol et al. \cite{arrepol2023separating}. In particular, Foucaud and Kov{\v{s}}e \cite{DBLP:journals/jda/FoucaudK13} show that $f(T) \le \frac{2n}{3} + O(1)$, for any $n$-vertex tree $T$. Moreover, they show that for the star $K_{1,n-1}$, $f(K_{1, n-1}) = \frac{2n}{3} + O(1)$. Thus $K_{1,n-1}$ serves as a matching lower bound example. Arrepol et al. \cite{arrepol2023separating} provided bounds as a function of the number of degree one vertices and degree two vertices in $T$, denoted by $A_1$ and $A_2$ respectively, and the number of special bare paths $\mathcal{I}$ --- the number of paths of length at least two in $T$ such that its two endpoints have a degree at least three and all the other vertices on the path have degree two. Their bound for trees reads as follows:

\begin{thm}(\cite[Theorem 3.4]{arrepol2023separating}) \label{best_known}
    Let $T$ be a tree with $A_1$ degree one vertices, $A_2$ degree two vertices, and $\mathcal{I}$ special bare paths. Then,
\[ \max \left\{\left\lceil \frac{2A_1 + A_2 - \mathcal{I}}{3} \right\rceil, \left\lceil \frac{A_1 + A_2 - \mathcal{I}}{2} \right\rceil \right\} \ \le \ f(T) \ \le \ \frac{2A_1}{3} + \frac{A_2 - \mathcal{I}}{2} + O(1).\]
\end{thm}

In this section, we present a tight lower bound for the size of an optimal vertex-separating path system of any $n$-vertex tree as a function of $n$. Our lower bound does not follow directly from \cref{best_known}. For example, one can check that if $T$ is an $n$-vertex tree obtained from a binary tree where every edge between two vertices of degree three is subdivided once, then \cref{best_known} implies that $f(T) \ge \frac{2n}{9}$.

\begin{thm} \label{trees_label}
    Let $T$ be an $n$-vertex tree. Then $f(T) \ge \frac{n}{4}$. Moreover, this lower bound is tight (up to an additive constant), since there are infinitely many trees $T$ with $f(T) \le \frac{|V(T)|}{4} + O(1)$.
\end{thm}

We start with a simple observation relating the low-degree vertices of a graph with the endpoint of paths in an arbitrary vertex-separating path system. 

 \begin{obs} \label{ez}
    Let $G$ be a graph and let $\mathcal{S}$ be a vertex-separating path system of $G$.
    \begin{enumerate}
        \item[(I)] Let $A_1$ be the number of degree one vertices in $G$. At least $A_1 - 1$ degree one vertices are the endpoints of some path in $\mathcal{S}$.
        \item[(II)] Let $\Pi \in \mathcal{S}$ be a non-trivial path such that both endpoints, say  $u$ and $v$, of $\Pi$ are degree one vertices of $G$. There exists a path $\Pi' \in \mathcal{S}$, $\Pi'\not=\Pi$, such that $\Pi'$ contains exactly one of $u$ and $v$ as an endpoint.
        \item[(III)] Let $u$ and $v$ be two adjacent degree two vertices of $G$. There is a path in $\mathcal{S}$ that ends in exactly one of $u$ and $v$.
        \item[(IV)] Let $u$ be a degree one vertex of $G$ that is adjacent to a degree two vertex $v$. If there is no trivial path containing $u$ in $\mathcal{S}$, there exists a path in $\mathcal{S}$ with $v$ as an endpoint that does not contain $u$.
    \end{enumerate}
\end{obs}

We will use this observation to prove a tight lower bound for $f(T)$ where $T$ is an arbitrary $n$-vertex tree.

\begin{prop}
Let $n > 2$ be an integer and $T$ be an $n$-vertex tree. Then $f(T) \ge \frac{n}{4}$.
\end{prop}

\begin{proof} Let $T$ be an arbitrary $n$-vertex tree. Let $A_1$, $A_2$, and $A_{\ge3}$ be the number of degree one, degree two, and degree at least three vertices of $T$, respectively. Let $\mathcal{S}$ be a vertex-separating path system of $T$. If $T$ is a path then by \cref{cycle}, $|\mathcal{S}| = \lceil\frac{n}{2}\rceil \ge \frac{n}{4}$. Therefore, we can assume $T$ is rooted at a vertex of degree at least three. Let $\phi$ be the number of leaves of $T$ that create trivial paths in $\cal S$. We prove the lower bound by induction on pair $(\phi, n)$ in lexicographic order. For the base case, assume that $\phi = 0$.

We say that an edge $e = \{x, y\}$ in $T$ is \defin{good} if $x$ is a degree two vertex of $T$ and $y$ is the only child of $x$ such that $y$ is either a vertex of degree at most two or there is no path in $\cal S$ that contains $y$ but does not contain $x$.

\begin{clm}
    For each good edge $e \in E(T)$, there exists a distinct endpoint of some path in $\cal S$ located on a degree two vertex of $e$.
\end{clm}

\begin{proof}
Let $\{x,y\} \in E(T)$ be a good edge where $x$ is a vertex of degree two and $y$ is the child of $x$ in $T$. We have three cases:

\begin{enumerate}
    \item[(I)] $y$ is a vertex of degree one. By \cref{ez} (IV) and since there is no trivial path in $\cal S$ ($\phi = 0$), there is a path $\Pi \in \cal S$ such that one endpoint of $\Pi$ is on the degree two vertex $x$ and $\Pi$ does not contain $y$.

    \item[(II)] $y$ is a vertex of degree two. By \cref{ez} (III), there is a path $\Pi \in \mathcal{S}$ that ends in exactly one of $x$ and $y$; in fact, $\Pi$ separates $x$ and $y$.

    \item[(III)] $y$ is a vertex of degree at least three and every path of $\cal S$ that contains $y$ also contains $x$. Then, in order to separate $x$ and $y$, there exists a path in $\cal S$ that ends at $x$ not containing $y$.

\end{enumerate}

For a good edge $e \in E(T)$, according to these three cases, an endpoint of some path in $\cal S$ is at a degree two endpoint of $e$. Let $u$ be a vertex of degree two, then each non-trivial path with one endpoint at $u$ contains one of the neighbors of $u$. Therefore, an endpoint of a non-trivial path in $\cal S$ does not separate endpoints of two good edges. If a degree two vertex is a trivial path in $\cal S$, we consider two endpoints for that trivial path. Therefore, each good edge has a distinct corresponding endpoint of a path from $\cal S$.
\renewcommand{\qedsymbol}{\scalebox{1.3}{$\triangle$}}
\end{proof}

\begin{clm}
    The tree $T$ has at least $A_2 - |S|$ good edges.
\end{clm}

\begin{proof}
Let $u$ be a vertex of degree two, and $v$ be its only child in $T$. If $v$ has degree at most two, then the edge $\{u,v\}$ is a good edge. For each path $\Pi \in \cal S$, only the parent of one vertex in $V(\Pi)$ is not in $\Pi$. Hence, there are at most $|\cal S|$ degree-two vertices $u \in V(T)$ such that its child $v$ in $T$ has degree at least three, and there is a path in $\cal S$ that contains $v$ but does not contain $u$. Therefore, there are at least $A_2 - |S|$ good edges in $T$. \renewcommand{\qedsymbol}{\scalebox{1.3}{$\triangle$}}
\end{proof}

Let $L$ be the set of degree-one vertices in $T$, and note that $|L| = A_1$. Let $D$ be the set of degree two vertices $v \in V(T)$, for which there exists a path $\Pi$ in $\mathcal{S}$ that has $v$ as an endpoint.
    % By \cref{ez} (I), there is an endpoint of a path in $\cal S$ located on at least $A_1 - 1$ vertices in $L$. 
    % If $|L| < |D|$, then using $A_1 \ge 2 + A_{\ge3}$, we have  
    % \[|S| \ge \frac{A_1 - 1 + |D|}{2} \ge A_1 \ge \frac{n}{4}.\]
    % From now on, assume $|L| \ge |D|$. 
    We use a charging scheme to prove the lower bound. We assign a charge of $1$ per path in $\mathcal{S}$. Then we consider the following discharging rules. Let $\Pi \in \mathcal{S}$ be a path with $u$ and $v$ as endpoints.
        \begin{enumerate}
            \item[(I)] If $u, v \in L \cup D$, then both $u$ and $v$ get $\frac{1}{2}$ charge from $\Pi$.
            \item[(II)] If $u \in L \cup D$ and $v \notin L\cup D$, then $u$ gets $1$ charge from $\Pi$.
        \end{enumerate}
        
    After discharging, every vertex in $D$ gets at least $\frac{1}{2}$ charge. Denote by $ch_D$ the total charge stored at the vertices of $D$. Recall that for each good edge $e \in E(T)$, there exists at least one endpoint of some path in $\cal S$ located on a degree two vertex of $e$. Note that these endpoints are distinct. Therefore, $ch_D$ is at least $\frac{1}{2}(A_2 - |S|)$.
    
    Let $X$ be the set of all vertices in $L$ such that there exists a path $\Pi \in \cal S$ with one endpoint in $X$ and one endpoint in $D$. Let $r = |X|$. By definition of $X$, $\frac{r}{2}$ is another lower bound for $ch_D$, that is $ch_D \ge \frac{r}{2}$.
    Then by \cref{ez} (II, IV), at least $\frac{2}{3}(|L| - r - 1)$ charge is located on vertices in $L\setminus X$ (cf. \cite[Proposition 12]{DBLP:journals/jda/FoucaudK13}). Hence,
        \[|\mathcal{S}| \ge \frac{2}{3}(|L| - r - 1) + \frac{1}{2}r + ch_D = \frac{2}{3}A_1 - \frac{1}{6}r + ch_D - \frac{2}{3}.\ \ \ \ \ \ \ (\spadesuit) \]
        
    We consider two cases:
    \begin{enumerate}
        \item[(a)] $r \ge A_2 - |\mathcal{S}|$: By $(\spadesuit)$ and $ch_D \ge \frac{r}{2}$, we have
            \[|\mathcal{S}| \ge \frac{2}{3}A_1 - \frac{1}{6}r + ch_D - \frac{2}{3} \ge \frac{2}{3}A_1 + \frac{1}{3}r - \frac{2}{3} \ge \frac{2}{3}A_1 +\frac{1}{3} (A_2 - |\mathcal{S}|) - \frac{2}{3}. \]
        \item[(b)] $r < A_2 - |\mathcal{S}|$: By $(\spadesuit)$ and $ch_D \ge \frac{1}{2}(A_2 - |\mathcal{S}|)$, we have
            \[|\mathcal{S}| \ge \frac{2}{3}A_1 - \frac{1}{6}r + ch_D - \frac{2}{3} \ge \frac{2}{3}A_1 - \frac{1}{6} (A_2 - |\mathcal{S}|) + \frac{1}{2} (A_2 - |\mathcal{S}|) - \frac{2}{3} \ge \frac{2}{3}A_1 +\frac{1}{3} (A_2 - |\mathcal{S}|) - \frac{2}{3}. \]
    \end{enumerate} 
    In both cases, we have $|\mathcal{S}| \ge \frac{2}{3}A_1 +\frac{1}{3} (A_2 - |\mathcal{S}|) - \frac{2}{3}$, that is
        \[ |\mathcal{S}| \ge \frac{1}{2}A_1 + \frac{1}{4}A_2 - \frac{1}{2} = \frac{1}{4}A_1 + \frac{1}{4}A_2 + \frac{1}{4}A_{1} - \frac{1}{2} \ge \frac{1}{4}A_1 + \frac{1}{4}A_2 + \frac{1}{4}(A_{\ge3} + 2) - \frac{1}{2} = \frac{n}{4}.\]

    For the induction step, assume $\phi > 0$ and let $v \in V(T)$ be a leaf that creates a trivial path in $\cal S$. We consider two cases:
    
    Let $u \in V(T)$ be the only neighbor of $v$ in $T$. Assume $u$ creates a trivial path in $\cal S$ and the edge $(v, u)$ is a path in $\cal S$. Let $\cal S'$ be a set obtained by removing the trivial path corresponding to $v$ from $\cal S$. Observe that $\cal S'$ is a vertex-separating path system of $T$. By induction hypothesis, we have $|\mathcal{S}| > |\mathcal{S}'| \ge \frac{n}{4}$.
    
    Otherwise, let $T'$ be the tree that is obtained from $T$ by removing $v$. Let $\cal S'$ be a vertex-separating path system of $T'$ obtained from $\cal S$ by excluding trivial path at $v$ and removing $v$ from every path in $\cal S$. Since the number of trivial paths in $\cal S'$ is not greater than $\phi$ and the number of vertices of $T'$ is one less than the number of vertices of $T$, we have $|\mathcal{S}'| \ge \frac{n-1}{4}$. Therefore $|\mathcal{S}| \ge \frac{n-1}{4} + 1 > \frac{n}{4}$.
\end{proof}

In order to prove the tightness, we build infinitely many trees $T$ such that their vertices can be separated using $\frac{|V(T)|}{4} + O(1)$ paths.

\begin{prop}
    There are infinitely many trees $T$ with $f(T) = \frac{|V(T)|}{4} + O(1)$.
\end{prop}

\begin{proof}
Consider a planar drawing of the complete binary tree $B_h$ of height $h-1$, with $2^h - 1$ vertices, for a positive integer $h$. We label the leaves of $B_h$ with the numbers $0, 1, 2, \ldots, 2^{h-1}-1$ from left to right. We construct a vertex-separating path system, denoted by $\mathcal{S}$, for $B_h$. $\mathcal{S}$ contains the unique path between the leaves labeled $i$ and $i+1$, for each $i$ with $0 \le i < 2^{h-1}$ and $i \not\equiv 0$ (mod 4), and the unique path between the leaves labeled $0$ and $2^{h-1}-1$ (see \cref{n4trees}). If we connect the leaves labeled with consecutive numbers by an edge, we obtain a planar graph where each path corresponds to a face of this graph. Due to the properties of the complete binary tree, every pair of non-adjacent vertices in this planar graph is incident to at most one common internal face. Using a simple inductive argument, it can be shown that $\mathcal{S}$ is indeed a vertex-separating path system of $B_h$. Moreover, every pair of edges is separated from each other, meaning that every pair of distinct edges $e_1$ and $e_2$ in $E(T)$ are separated by a path in $\mathcal{S}$ (Except the edges incident to the root).

\begin{figure} [t]
  \begin{center}
    \includegraphics[scale=0.85]{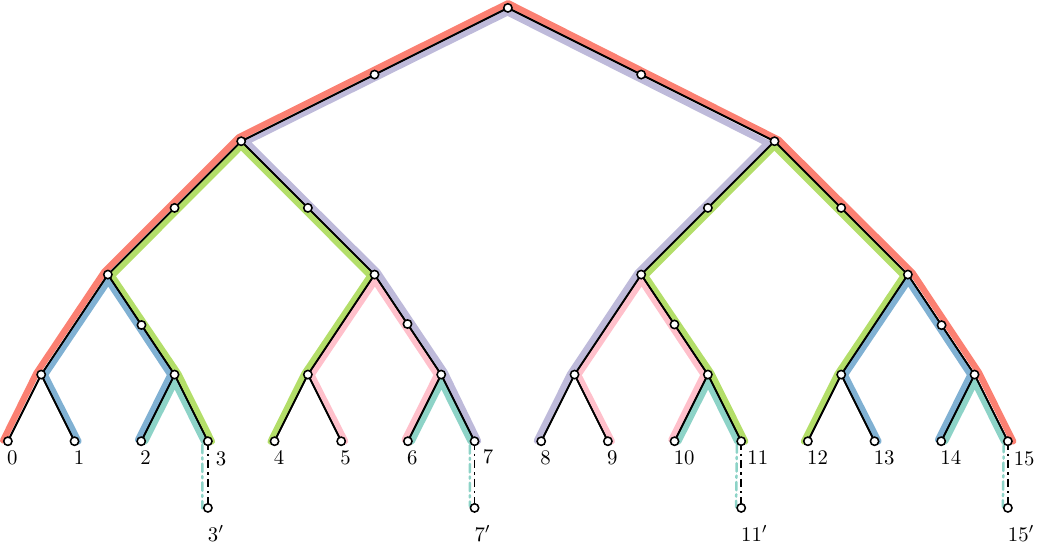}
  \end{center}
  \caption{The tree $T_5$ obtained from the complete binary tree of height $4$.}
  \label{n4trees}
\end{figure}

We number the levels in $B_h$ from $1$ to $h$, where the root is at level one. Next, by adding $2^{h-2}-2$ vertices, we subdivide every edge with both endpoints at the levels at most $h-2$. Then, consider all the edges with one endpoint at level $h-2$ and another endpoint at level $h-1$ in the left to right order. We subdivide every second one of such edges. This step will add $2^{h-3}$ vertices to the graph (see \cref{n4trees}). We refer to the tree obtained after these modifications, with $2^{h} + 2^{h-2} + 2^{h-3} - 3$ vertices, as $B'_h$. Let $\mathcal{S'}$ be the set of paths in $B'_h$ with the same endpoints as the paths in $\mathcal{S}$. To show that $\mathcal{S'}$ is a vertex-separating path system of $B'_h$, we only need to consider the separation of subdivision vertices (the separation from the root of $B'_h$ is a special case and will be handled separately). Every subdivision vertex $v$ is on two different paths in $\mathcal{S'}$ and is incident to two vertices of degree three, say $v_1$ and $v_2$. Therefore, by construction, there are three different paths going through each neighbor of $v$, while exactly two of them pass through $v$. To separate $v$ from other vertices in $V(B'_h)\setminus \{v_1,v_2\}$, we consider the cases where $u$ is a subdivision vertex, an internal vertex of $B_h$, or a leaf of $B'_h$. In the former two cases, we use the fact that there is an edge incident to $u$ that is separated by a path in $\mathcal{S}$ from the edge on which $v$ is located. For the latter case, we note that at least one of the two paths that go through the edge $\{v_1, v_2\}$ in $B_h$ does not cover the vertex $u$.

For each $0 \le \ell < 2^{h-1}$ where $\ell$ mod $4 = 2$, according to our construction, there exists a path of length two between the leaf labeled $\ell$ and $\ell+1$. As the final step of the construction, we create the tree $T_h$ by adding an edge to each leaf located at $\ell+1$. We extend the path of length two between $\ell$ and $\ell+1$ in $\mathcal{S'}$ to cover the new edge (indicated by dashed edges in \cref{n4trees}). Once again, based on our construction, the extended paths of length three will separate the new vertex from the rest of the vertices. This final step adds $2^{h-3}$ vertices to obtain the final tree, resulting in a total of $3 \cdot 2^{h-1} - 3$ vertices in $T_h$. As mentioned before, the only vertex that is not separated from its two neighbors is the root of the binary tree, and this can be solved by adding two extra paths.

Now if $\mathcal{P}$ denotes our vertex-separating path system of $T_h$, then we have
\[ |\mathcal{P}| \le 3 \cdot 2^{h-3} + 2 < \frac{1}{4} \cdot |V(T_h)| + 3 = \frac{1}{4} \cdot |V(T_h)| + O(1).\]
\end{proof}

By combining the previous two propositions, one can establish the proof of \cref{trees_label}.

\section{Grid Graphs} \label{grid_sec}

In this section, we study vertex-separating path systems of grids. Vertex-separating cycle systems of the $m\times n$ torus have been studied in \cite{DBLP:journals/dam/HonkalaKL03} and \cite{DBLP:journals/combinatorics/Rosendahl03}. Here, we separate the vertices of an $m\times n$ grid, $G_{m,n}$, using paths. We start with a simple observation about the existence of a Hamiltonian path in grid graphs.

\begin{obs} \label{ham_grid}
    Let $m$ and $n$ be two positive integers. The grid $G_{m,n}$ has a Hamiltonian path with both endpoints on the last row.
\end{obs}

Now we state the main theorem of this section regarding the value of $f(G_{m,n})$ when $m, n \ge 2$.

\begin{thm} \label{grid}
Let $m, n \ge 2$ be two integers. Then $f(G_{m,n}) \le 2\lceil \log{m} \rceil + 2\lceil \log{n} \rceil$.
\end{thm}

\begin{proof}
    Let $m, n \ge 2$ be two integers. Recall that we represent the vertex set of $G_{m,n}$ by $V(G_{m,n}) = \{0,1,2, \ldots, m-1\}\times\{0,1,2, \ldots, n-1\}$. To begin with, we separate the vertices within different columns of the subgrid induced by the first $m-1$ rows of $G_{m,n}$ from each other by a separating path system of size $O(\log{n})$. With a similar idea as in \cref{kn}, we find a nice labeling of the columns. For each $2 \le i \le \lceil \log{n} \rceil$, let $A_i = \{(x, y) \in V(G_{m,n})$ $|$ $x < m-1$ and the $i$-th bit in the binary representation of $y$ is $1\}$. We construct a spanning path of $A_i$, say $\Pi_i$, such that $\Pi_i$ does not intersect the first $m-1$ rows of $G_{m,n}$ on vertices other than vertices in $A_i$.
    
    Note that, for each $2 \le i \le \lceil \log{n} \rceil$, each component of $G[A_i]$ consists of at least two consecutive columns of $G$. By \cref{ham_grid}, we consider a Hamiltonian path with endpoints on the $(m-1)$-th row for each component of $G[A_i]$. We can merge the Hamiltonian path of these sub-grids using the $m$-th row. Therefore, $\Pi_i$ as a spanning path of $A_i$ exists (e.g. path $\Pi_3$ is depicted in \cref{grid_fig}). By construction, this set of paths will separate the vertices of every pair of columns, except for $\lfloor \frac{n}{2} \rfloor$ pairs of columns, namely $\{0, 1\}, \{2, 3\}, \{4, 5\}, \ldots, \{2\lceil \frac{n}{2} \rceil - 2, 2\lceil \frac{n}{2} \rceil - 1\}$.
    
    For the purpose of separating these pairs of columns, we introduce the set $A_1 = \{(x, y) \in V(G_{m,n}) | x < m-1$ and $y$ mod $4 = 1 $ or $2\}$. We then construct a spanning path $\Pi_1$ for the vertices of $A_1$ following the same procedure as described previously for $\Pi_2, \ldots, \Pi_{\lceil \log{n} \rceil}$ (refer to \cref{grid_fig}). Overall, the paths $\Pi_1, \Pi_2, \Pi_3, \ldots, \Pi_{\lceil \log{n} \rceil}$ separate every pair of vertices located in the first $m-1$ rows and in different columns. By repeating the same idea for the last $m-1$ rows, we can guarantee that with $2\lceil \log{n} \rceil$ paths, every pair of vertices $u, v \in V(G_{m,n})$ is separated if $\pi_y(u) \neq \pi_y(v)$.

    \begin{figure}[t]
    \centering
    \includegraphics[width=0.97\textwidth]{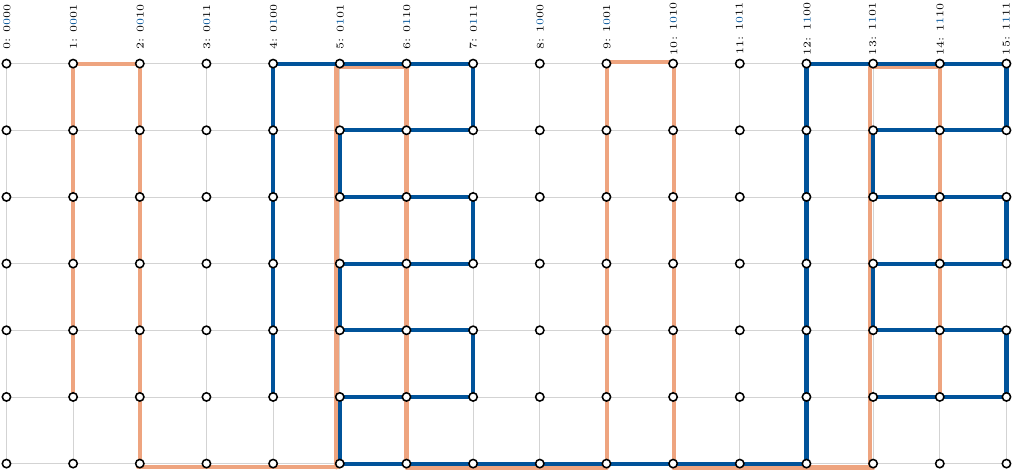} %width=0.95\textwidth
    \caption{Path $\Pi_3$ in blue colour and path $\Pi_1$ in orange colour in a $7\times16$ grid}
    \label{grid_fig}
    \end{figure}

    Using an analogous construction, a set of paths with a size of $2\lceil \log{m} \rceil$ would separate vertices located in different rows. Since every pair of vertices in $V(G_{m,n})$ differs in their $\pi_x$ or/and $\pi_y$ values, these paths will form a vertex-separating path system of size at most $2\lceil \log{m} \rceil + 2\lceil \log{n} \rceil$.
\end{proof}

\section{Maximal Outerplanar Graphs} \label{outer}

In this section, we explore vertex-separating path systems of maximal outerplanar graphs, i.e., graphs that are triangulations of convex polygons. The \defin{inner dual} of a maximal outerplanar graph is its dual where the vertex corresponding to the outer face is removed. \cref{grid} implies that $f(G_{2,n}) = O(\log{n})$. Hence, unlike trees, there are outerplanar graphs with a vertex-separating path system of size $O(\log{n})$. The class of trees is a subclass of outerplanar graphs. Therefore, we cannot hope to improve the upper bound for the class of outerplanar graphs in general. However, since each 2-connected outerplanar graph $G$ has a Hamiltonian cycle, using \cref{cycle}, one would notice that $f(G) \le \lceil \frac{n}{2} \rceil$. In fact, considering maximal outerplanar graphs leads to a further improvement in the upper bound. As a building block of maximal outerplanar graphs, we first consider the fan graph, $F_n$, on $n$ vertices, that is, a path $P_{n-1}$ on $n-1$ vertices and an apex vertex connected to all the vertices of the path. %The proof of the following lemma is in \cref{fan_proof}.

\begin{lem} \label{fan}
    Let $F_n$ be the fan graph with $n$ vertices. Then $f(F_n) = n/4 + O(1)$.
\end{lem}

\begin{proof} 
    Let $\mathcal{P}$ be a vertex-separating path system of $F_n$. Using $\mathcal{P}$, we obtain a vertex-separating path system, $\mathcal{P}'$, for $P_{n-1}$. For each path $\Pi \in \mathcal{P}$, if $\Pi$ contains the apex vertex of $F_n$ then we create at most two paths for $\mathcal{P}'$ from $\Pi$ by removing the apex vertex, otherwise, we include $\Pi$ in $\mathcal{P}'$. Since $\mathcal{P}$ is a separating path system for $F_n$, $\mathcal{P}'$ is a separating path system for $P_{n-1}$. By \cref{cycle}, we know $\lceil\frac{n-1}{2} \rceil=f(P_{n-1}) \le 2\cdot f(F_n)$, therefore, $f(F_n) \ge \frac{n-1}{4}$.

    To prove the upper bound, we construct a separating path system of size $\frac{n + 10}{4}$ when $n \ge 6$. Let $\Pi$ be the induced path of size $n-1$ in $F_n$. We split $\Pi$ into sub-paths $\Pi_l$ and $\Pi_r$ of sizes $\left\lfloor \frac{n-1}{2} \right\rfloor$ and $\left\lceil \frac{n-1}{2} \right\rceil$, respectively. In order to separate the vertices of $\Pi_l$ and $\Pi_r$, we include the path $\Pi_l$ in the separating path system. Again by \cref{cycle}, we construct a separating path systems $\mathcal{P}_l$ and $\mathcal{P}_r$ of sizes at most $\left\lceil \frac{|\Pi_r|}{2} \right\rceil \le \frac{n+2}{4}$ for $\Pi_l$ and $\Pi_r$, respectively. Since every pair of vertices $u \in \Pi_l$ and $v \in \Pi_r$ are already separated, we can merge the paths of $\mathcal{P}_l$ and $\mathcal{P}_r$ through the apex of $F_n$. Moreover, we add a single vertex path on the apex to separate the apex form the rest of the vertices in $V(F_n)$. By construction and the fact that $n \ge 6$, observe that this set of paths is a vertex-separating path system for $F_n$. Hence, $f(F_n) \le \frac{n+2}{4} + 2 \le \frac{n + 10}{4}$.
\end{proof}

In \cref{poly-lowerbound}, we  confirm that, indeed, having just one high-degree vertex in an outerplanar graph is sufficient to establish a polynomial lower bound for the size of a vertex-separating path system.
To extend the result of \cref{fan} to all maximal outerplanar graphs, we make the following observation.

\begin{obs} \label{fan_dec}
    Let $G$ be a maximal outerplanar graph such that the inner dual of $G$ is a path. Then $G$ can be decomposed into maximal induced fan subgraphs, $F_1, F_2, F_3, \ldots, F_k$, such that $G = \bigcup^{k}_{i=1} F_i$ and $|V(F_i) \cap V(F_j)| \le 2$ for $1 \le i < j \le k$.
\end{obs}

\begin{lem} \label{inner2}
    Let $G$ be an $n$-vertex maximal outerplanar graph such that the inner dual of $G$ is a path, then $f(G) \le \frac{n}{4} + O(1)$.
\end{lem}

\begin{proof}
    Let $\mathcal{F} := \{F_1, F_2, F_3, \ldots, F_k\}$ be the set of all maximal fan subgraphs of $G$, as in \cref{fan_dec}. Assume $k > 1$, otherwise, the result is implied from \cref{fan}. For each $1 \le i \le k$, the vertex set of the graph $F_i$ consists of an apex vertex $a_i$ and vertices, $\mathcal{V}_i := \{v_{i,1}, v_{i,2}, v_{i, 3}, \ldots, v_{i, n_i}\}$, along an induced path on the outer face of $G$ in clockwise order. Note that $2 \le n_i = |V(F_i)|-1$ and $a_i$ is the unique common neighbour of the vertices in $\mathcal{V}_i$. Moreover, if $1 \le i < k$, $v_{i, n_i} = v_{(i+1), 1}$.

    \begin{figure} [t]
        \begin{center}
          \includegraphics[scale=0.66]{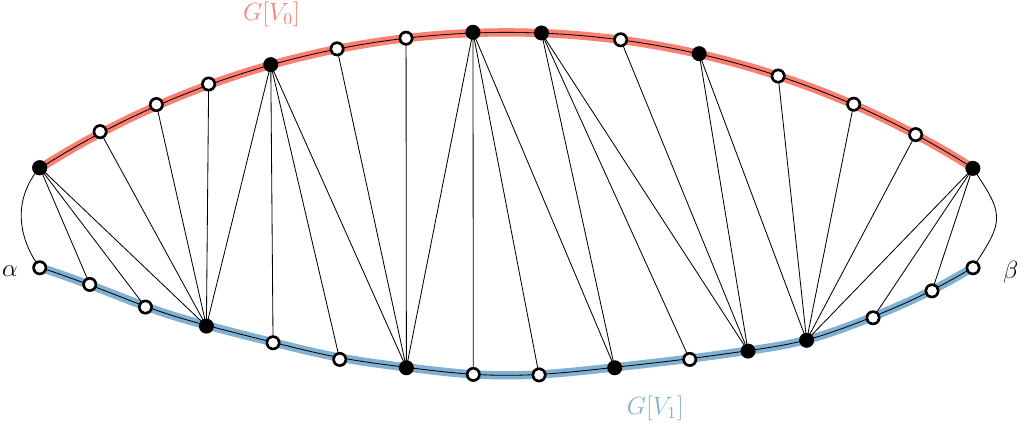}
        \end{center}
        \caption{An outerplanar graph in which the inner dual is a path. The black vertices represent the apex vertices of maximal fan subgraphs.}
        \label{pathinner}
    \end{figure}

    Since the inner dual of $G$ is a path, $G$ contains exactly two vertices of degree two, say $\alpha$ and $\beta$. There are two different paths, $P_{\alpha\beta}$ and $P_{\beta\alpha}$, between $\alpha$ and $\beta$ on the outer face of $G$. None of the vertices $\alpha$ and $\beta$ is an apex vertex for any graph in $\mathcal{F}$ (See \cref{pathinner}). We partition the set $\mathcal{F}$ into two sets of almost equal size $\mathcal{F}_0 = \{F_i \in \mathcal{F} | a_i \in P_{\alpha\beta}\}$ and $\mathcal{F}_1 = \{F_i \in \mathcal{F} | a_i \in P_{\beta\alpha}\}$ and define $V_0 := \bigcup_{F_i \in \mathcal{F}_0} \mathcal{V}_i$ and $V_1 := \bigcup_{F_i \in \mathcal{F}_1} \mathcal{V}_i$. $V_0$ and $V_1$ create a partition of $V(G)$ and $G[V_0]$ and $G[V_1]$ are two disjoint paths on the outer face of $G$. 

    For each $i \in \{0, 1\}$, let $G_i$ be the fan graph obtained by contracting the connected subgraph $G[V_i]$ of $G$ into a single vertex. We separate the vertices within $V_{i}$ from each other by applying the result of \cref{fan} to $G_{1-i}$, where the subgraph $G[V_{1-i}]$ plays the role of the apex vertex in the construction explained in \cref{fan}. To separate the vertices of $V_0$ and $V_1$ from each other, we only need to consider one extra path, namely $G[V_0]$ or $G[V_1]$\footnote{In the proof of \cref{outerthm}, we include both $G[V_0]$ and $G[V_1]$ in the vertex-separating path system. This guarantees that the resulting separating system covers $V(G)$.}. So the total number of paths used to separate $V(G)$ is at most $f(G_0) + f(G_1) + 1 \le \frac{n}{4} + O(1)$.
\end{proof}

\begin{thm} \label{outerthm}
    Let $G$ be an $n$-vertex maximal outerplanar graph for $n > 3$. Then $f(G) \le \frac{n}{4} + O(1)$. Moreover, this upper bound is tight up to an additive constant.
\end{thm}

% \begin{figure} [t]
%     \begin{center}
%       \includegraphics[scale=0.9]{vsp/figs/outerbinary.pdf}
%     \end{center}
%     \caption{Decompostion of outerplanar graph. {\color{orange} i don't like this figure!}}
%     \label{outerbinary}
% \end{figure}

\begin{proof}
    Let $n$ be the number of vertices of $G$ and $k$ be the number of leaves of the inner dual of $G$. We prove this statement by induction on pair $(k, n)$. For the base case, note that the inner dual of each maximal outerplanar graph contains at least two leaves and \cref{inner2} proves the statement when the number of leaves is equal to two.
    
    For the inductive step, if the inner dual of $G$ has at most $k$ leaves, where $3 \le k \le 19$, then we decompose $G$ into constant number of maximal outerplanar graphs $G_1, G_2, G_3, \ldots, G_{k-1}$, such that $G = \bigcup^{k-1}_{i = 1} G_i$ and the inner dual of $G_i$ is a path and $|V(G_i) \cap V(G_j)| \le 1$ for $1 \le i < j < k$. For each $1 \le i < k$, let $\mathcal{P}_i$ be the separating path system obtained by applying \cref{inner2} to $G_i$. By construction, each $\mathcal{P}_i$ is a covering path system for $G_i$. Consider $\mathcal{P} = \bigcup^{k-1}_{i = 1} \mathcal{P}_i$ as a separating path system of $G$. Since $k$ is a constant, we have 
    \[f(G) \le \sum^{k-1}_{i = 1} f(G_i) \le \sum^{k-1}_{i = 1} \left(\frac{V(G_i)}{4} + O(1)\right) \le \frac{n}{4} + O(1).\]
    Now assume that the inner dual of $G$ has $k \ge 20$ leaves. We use at most $\lceil \log(k+1) \rceil$ paths to separate the degree two vertices of $G$ associated with these $k$ leaves of the dual. Let $S$ be the set of degree two vertices of $G$; thus, $|S| = k$. We are aiming to find a nice labeling for the vertices of $S$. The graph $G' = G[V(G)\setminus S]$ is a maximal outerplanar graph. We show that for each $A \subseteq S$, there exists a path in $G$ that covers all vertices of $A$ and no vertex of $S \setminus A$. Because no two vertices of degree two of $G$ are adjacent, every degree two vertex of $G$ is adjacent to two consecutive vertices on the outer face of $G'$. We construct such a path by considering a Hamiltonian path of $G'$ and extending it to contain only the vertices of $A$. Therefore, we can apply the result of \cref{kn}, to separate the vertices in $S$ from each other by a separating path system of size $\lceil \log(k+1) \rceil$ (the plus one in the log function is to ensure that the constructed separating path system covers $S$). Observe that $G'$ is a maximal outerplanar graph with a strictly smaller number of vertices and the number of leaves in its inner dual is no more than the number of leaves in the inner dual of $G$. We apply the induction hypothesis to the graph $G'$. Since $\frac{\lceil \log(k+1) \rceil}{k} \le \frac{1}{4}$ for $k \ge 20$, we have
    \[f(G) \le f(G') + \lceil \log(k+1) \rceil \le \frac{|V(G')|}{4} + O(1) + \lceil \log(|S|+1) \rceil\]
    \[ \le \frac{|V(G)|}{4} + O(1) = \frac{n}{4} + O(1).\]
    \cref{fan} proves the tightness of this upper bound up to an additive constant.
\end{proof}

\section{Graph Classes With Polynomial Lower Bound} \label{poly-lowerbound}

This section will demonstrate that for every positive integer $t$, $K_{2,t}$-minor-free graphs with a high-degree vertex require a polynomial-sized vertex-separating path system.

To begin with, we introduce the Vapnik–Chervonenkis (VC) dimension of a set system and the relevant concepts. Let $(X, \mathcal{S})$ be a set system. We say that a subset $A \subseteq X$ is \defin{shattered} by $\mathcal{S}$ if every subset of $A$ can be expressed as the intersection of some $B \in \mathcal{S}$ with $A$. We define the \defin{VC-dimension} of $(X, \mathcal{S})$ as the supremum of the sizes of all finite subsets of $X$ that can be shattered by $\mathcal{S}$. We define the \defin{shatter function} of a set system $(X, \mathcal{S})$ as:
\[\pi_\mathcal{S}(k) = \max_{Y \subseteq X, |Y| = k} |\{B \cap Y\ |\  B \in \mathcal{S}\}|.\]
In words, $\pi_\mathcal{S}(k)$ is the maximum possible number of distinct intersections of the sets of $\mathcal{S}$ with a $k$-element subset $Y$ of $X$. It is bounded by the following lemma:

\begin{lem}[\cite{VC71,Sau72,She72,matousek2013lectures}]\label{lm:sauer}
  Let $(X, \mathcal{S})$ be a set system with VC-dimension $d$ then $\pi_{{\cal S}}(k) \leq \Phi_d(k)$, where $\Phi_d(k) = \binom{k}{0}+\binom{k}{1}+ \dots +\binom{k}{d}$. In particular, for $k \geq d$ one has $\pi_{{\cal S}}(k) \leq (\frac{e} d)^d \cdot k^d$, where $e$ is Euler's number. 
\end{lem}

The \defin{dual set system} of $(X, \mathcal{S})$ is the set system $(\mathcal{S}, {\cal X}^*)$, where ${\cal X}^* := \{\mathcal{S}_x | x \in X\}$ and $\mathcal{S}_x := \{B \in \mathcal{S} | x \in B \}$. The \defin{dual shatter function} of the set system $(X, \mathcal{S})$ is the shatter function of the dual set system of $(X, \mathcal{S})$ and is represented by $\pi^*_{\mathcal{S}}(k)$. For a set system $(X,{\cal S})$, a set $B\in {\cal S}$ \defin{crosses} a pair of elements $\{x,y\}\subseteq {X}$ if and only if exactly one element in $\{x,y\}$ is contained in $B$. In the terminology of this paper, we say $B$ separates $\{x, y\}$. A \defin{matching} on the set $X$ is a disjoint collection of pairs of elements of $X$. A \defin{perfect matching} of $X$ is a matching of size $\left \lfloor \frac{|X|}{2}\right \rfloor$. We define the \defin{crossing number} of a matching $M$ on the elements in $X$ with respect to ${\cal S}$ to be the maximum number of pairs in $M$ crossed by any set $B \in {\cal S}$. The following result illustrates the relationship between the dual shatter function and the crossing number of a matching on a set system.

\begin{thm}[\cite{Chazelle1989},\cite{HAUSSLER1995217}] \label{thm:crossing}
    Let $(X,{\cal S})$ be a set system with $|X|=n$ and dual shatter function $\pi^*_{{\cal S}}(k)=O(k^d)$. Then there exists a perfect matching on elements of $X$ with crossing number $O(n^{1-1/d})$ with respect to ${\cal S}$.
\end{thm}

Now we are ready to state the main theorem of this section regarding a graph class with a polynomial-sized vertex-separating path system. %We start by proving the following key lemma, which plays a central role in the proof of the claim.

\begin{thm} \label{poly-lower}
    Let $t > 0$ be an integer and $G$ be a $K_{2,t}$-minor-free graph. Let $v_0 \in V(G)$ be a vertex of degree $\Omega(n^{\delta})$. Then there exists an $\epsilon := \epsilon(t, \delta)=\frac{\delta}{t+1}$ such that for any vertex-separating path system $\mathcal{P}$ of $G$, we have $|\mathcal{P}| = \Omega(n^\epsilon)$.
\end{thm}

\begin{proof}
    Let $\mathcal{P}$ be a separating path system of $G$. Define the set $\mathcal{P}'$ using $\mathcal{P}$ as follows. For each path $\Pi \in \mathcal{P}$, if $v_0 \in V(\Pi)$ then we include in $\mathcal{P}'$ at most two subpaths of $\Pi$ created by removing $v_0$ from $\Pi$; otherwise, we include $\Pi$ in $\mathcal{P}'$. Note that $|\mathcal{P}'| \le 2 \cdot |\mathcal{P}|$.
    
    Consider the set system $(X, \mathcal{S})$ where $X := N(v_0)$ and $\mathcal{S} := \{X \cap \Pi \ |\  \Pi \in \mathcal{P}' \}$. Let $({\cal S},{\cal X}^*)$ be the dual set system to $(X, \mathcal{S})$. 
    We claim that the set system $({\cal S},{\cal X}^*)$ has a VC-dimension at most $t+1$. Assume to the contrary that $({\cal S},{\cal X}^*)$ has a VC-dimension of at least $t+2$. Let ${\cal Q} := \{\Pi_1, \Pi_2, \dots, \Pi_{t+1}, \Pi_{t+2}\}$ be a set of $t+2$ distinct paths in ${\cal S}$ which are shattered by ${\cal X}^*$.
    By definition, for each $1 \le i \le t$, there are distinct vertices $\{x_1,x_2,\ldots,x_t,y_1,y_2,\ldots t_t\}$ such that $x_i\in \Pi_{t+1}\cap \Pi_i$ and $x_i$ is not contained in any other paths in ${\cal Q}$. Similarly, $y_i\in \Pi_{t+2}\cap \Pi_i$ and $y_i$ is not contained in any other path in ${\cal Q}$. For each $y_i$, let $y'_i$ be the neighbor of $y_i$ on the path $\Pi_i$ we meet by traversing it from $x_i$ to $y_i$. Let $\Pi'_i$ be the sub-path of $\Pi_i$ from $x_i$ to $y'_i$. Let $H$ be the graph obtained by contracting $\Pi_{t+1}$ and all of the $\Pi'_i$ to a vertex. By deleting some further edges from $H$ we obtain a star with the leaves $y_1,y_2,\ldots,y_t$. Hence $G$ contains a minor of $K_{2,t}$, a contradiction.

    By \cref{lm:sauer} and \cref{thm:crossing}, there is a matching $M$ of elements in ${X}$ such that each set $B \in {\cal S}$ crosses and therefore separates $O(n^{\delta(1-1/(t+1))})$ pairs of elements in $M$. Hence to separate every pair in $X$, we need $\Omega(n^{\delta}/n^{\delta(1-1/(t+1))})=\Omega(n^{\delta/(t+1)})$ paths in ${\cal P}'$. This implies that any vertex-separating path system of $G$ must have size $\Omega(n^{\delta/(t+1)})$.
\end{proof}

\begin{rem}
    The dual of a set system with a bounded VC-dimension has a bounded VC-dimension \cite{AIF_1983__33_3_233_0}. Using this fact and \cref{lm:sauer}, for a set system $({X},{\cal S})$ of VC-dimension $d$, we have that $\pi^*_{{\cal S}}(k)=O(k^{2^{d+1}})$. Therefore, in the proof of \cref{poly-lower}, to obtain an upper bound for the VC-dimension of the dual set system $({\cal S},{\cal X}^*)$, we might only rely on the upper bound of the VC-dimension of the primal set system $(X, \mathcal{S})$. It is not hard to see that the VC-dimension of the primal set system is also at most $t+1$. However, to obtain an improved lower bound for the size of the vertex-separating path system of $G$, we directly study the VC-dimension of the dual set system.
\end{rem}

As a corollary of this result, every outerplanar graph with a high degree (polynomial-sized) vertex requires a polynomial-sized vertex-separating path system. On another note, the result of \cref{poly-lower} can be generalized to $K_{3,t}$-minor-free graphs in a natural way. In fact, for a positive integer $t$, if $G$ is a $K_{3,t}$-minor-free graph with two vertices $u_0$ and $v_0$, and the common neighborhood of $u_0$ and $v_0$ has a size of $\Omega(n^{\delta})$, then there exists an $\epsilon := \epsilon(t, \delta)$ such that any vertex-separating path system of $G$ is required to have a size $\Omega(n^\epsilon)$. This, in turn, implies that every bounded genus graph with two vertices having a polynomial-sized common neighborhood requires a vertex-separating path system of polynomial size \cite{ringela, ringelb, DBLP:journals/jct/Bouchet78}.

\section{Separating Tree Systems} \label{tree_sep}

In this section, as a generalization of separating path systems, we study separating tree systems for the vertex set of a graph $G$. Recall that $f_t(G)$ is the size of a smallest vertex separating tree system of the graph $G$, that is, the smallest number of subtrees, $T_1, T_2, T_3, \ldots, T_k$, of $G$, such that each pair of vertices is separated by one of these trees. As a comparison, notice that in contrast to separating path systems, $f_t(K_{m,n}) = O(\log(n+m))$ \cite{DBLP:journals/combinatorics/Rosendahl03}.

\begin{thm}
    Let $T$ be an $n$-vertex tree. Then $\log(n) \le f_t(T) \le n / 2 + \log(n) + O(1)$, and these bounds are tight (up to lower order terms).
\end{thm}
    
\begin{proof}
    Let $v^*$ be a centroid of $T$, that is, a vertex of $T$, such that its removal results in connected components, each of size at most $\frac{n}{2}$. We root each component of $T \setminus \{v^*\}$ at a neighbor of $v^*$. First, we separate the vertices of different components from each other. With an idea similar to \cref{kn}, we assign a binary string of length $\lceil \log{(deg(v^*))} \rceil \le \log{n} + 1$ to each component of $T \setminus \{v^*\}$ and for each bit position we consider a spanning tree of $v^*$ and all components whose bit is $1$ at this position.

    Within each component of $T \setminus \{v^*\}$, we consider the paths from the root of the component to each of its vertices. Note that by the property of centroid, there are at most $\frac{n}{2}$ such paths in each component. We order these paths within each component arbitrarily. Using the vertex $v^*$ as a common neighbor, we merge the paths with the same index into one tree. This set of $\frac{n}{2}$ trees will separate the vertices located in the same component. To separate $v^*$ from all the other vertices, we add one single-vertex tree, covering $v^*$. In total we have at most $\frac{n}{2} + \log{n} + O(1)$ trees.

    The $\frac{n}{2}$ term in the upper bound is the best possible, because if the tree $T$ is a path, then a vertex-separating tree system is equivalent to a vertex-separating path system and, by \cref{cycle}, we know that $f_t(P_n) = \lceil \frac{n}{2}\rceil$. By \cref{kn}, the lower bound is trivial, and observe that by the above construction, one can see that $f_t(K_{1, n-1}) = \log{n} + O(1)$.
\end{proof}

In the rest of this section, we prove an upper bound on $f_t(T)$ in terms of the number of vertices and the radius of the tree $T$. For a non-empty subset $S$ of non-negative integers, define $b(S)$ to be the bitwise OR of the elements of $S$. If $k$ is the smallest integer such that $\max(S) < 2^k$, define $c(S) := \{2^k-1-x | x \in S\}$. In other words, the elements of $c(S)$ are obtained from the elements of $S$ by flipping the first $k$ bits in their binary representation.
    
\begin{lem} \label{bit-or}
    Let $[l_1, r_1]$ and $[l_2, r_2]$ be two disjoint intervals of integers, where, $0 \le l_1 \le r_1 < l_2 \le r_2$. Then $b([l_1, r_1]) \neq b([l_2, r_2])$ or $b(c([l_1, r_1])) \neq b(c([l_2, r_2]))$.
\end{lem}

\begin{proof}
    Let $k$ be the smallest integer such that $r_2 < 2^k$. We prove the lemma by induction on $k$. If $k = 1$, we have $l_1 = r_1 = 0$ and $l_2 = r_2 = 1$, therefore, $b([l_1, r_1]) \neq b([l_2, r_2])$. Let $k \ge 2$ and assume that the statement is true for all values less than $k$. In order to prove the statement for $k$, we consider three cases:
    \begin{itemize}
        \item $r_1 < 2^{k-1}$: Then the $k$-th bit in $b([l_1, r_1])$ is $0$. While, by definition, the $k$-th bit in $b([l_2, r_2])$ is $1$, we have $b([l_1, r_1]) \neq b([l_2, r_2])$.
        \item $r_1 \ge 2^{k-1}$ and $l_1 < 2^{k-1}$: Since $l_2 \ge 2^{k-1}$, the $k$-th bit in $b(c([l_2, r_2])$ is $0$. Since $l_1 < 2^{k-1}$, the $k$-th bit in $b(c([l_1, r_1])$ is $1$. Thus, $b(c([l_1, r_1])) \neq b(c([l_2, r_2]))$.
        \item $l_1 \ge 2^{k-1}$: In this case, the $k$-th bit is $1$ in the binary representation of every integer in $[l_1, r_1] \cup [l_2, r_2]$. In consequence, $b([l_i, r_i]) = b([l_i-2^{k-1}, r_i-2^{k-1}]) + 2^{k-1}$ and $b(c([l_i, r_i])) = b(c([l_i-2^{k-1}, r_i-2^{k-1}]))$ for $i \in \{1, 2\}$. Therefore, we only need to compare the values for the intervals $[l_1-2^{k-1}, r_1-2^{k-1}]$ and $[l_2-2^{k-1}, r_2-2^{k-1}]$, which implies the statement by the induction hypothesis.
    \end{itemize}
\end{proof}

\begin{thm}
    Let $T$ be an $n$-vertex tree with radius $r$. Then $max{(r, \log{n})} \le f_t(T) \le r + 2 \lceil \log{n} \rceil + 1$.
\end{thm}

\begin{proof}
We start with proving the lower bound. Let $\mathcal{T}$ be a vertex-separating tree system of $T$. Let $D$ be a longest path in $T$; note that $2r-1 \le |V(D)|$. The intersection of each element of $\mathcal{T}$ with $D$ is a path. We build a vertex-separating path system of $D$ using $\mathcal{T}' := \{D\cap t | t \in \mathcal{T} \}$. Now, using \cref{cycle}, we conclude that $r = \lceil\frac{2r-1}{2}\rceil \le \lceil\frac{|D|}{2} \rceil \le |\mathcal{T}|$. On the other hand, by \cref{kn}, we know that $\lceil \log{n} \rceil \le |\mathcal{T}|$. Therefore, $max{(r, \log{n})} \le f_t(T)$.

To prove the upper bound, let $c$ be a center of the tree $T$. We root the tree $T$ at $c$. Note that the height of $T$ is equal to $r$. We denote the set of leaves of the tree $T$ by $\ell(T)$. (If $c$ has degree one, then we do not consider $c$ to be a leaf.) We will construct two sets of trees.

The first set consists of $r+1$ trees. For $i=0,\ldots,r$, the tree $T_{i}$ is the subtree of $T$ consisting of all the vertices at distance at most $i$ from the root.

The second set consists of $2 \lceil \log | \ell(T) | \rceil$ trees. In order to define these trees, we label the leaves of $\ell(T)$ according to a post-order traversal of $T$ with the integers $0,1,\ldots, |\ell(T)|-1$.  
For each $1 \leq i \leq \lceil \log | \ell(T) | \rceil$: 
\begin{enumerate} 
\item Let $\Delta_{0,i}$ be the set of leaves in $\ell(T)$ whose labels
have a zero in the $i$-th position of their binary representations. 
\item Let $\Gamma_{0,i}$ be the smallest subtree of $T$ that is rooted 
at the root of $T$ and for which $\ell(\Gamma_{0,i}) = \Delta_{0,i}$.  
\item Let $\Delta_{1,i}$ be the set of leaves in $\ell(T)$ whose labels
have a one in the $i$-th position of their binary representations. 
\item Let $\Gamma_{1,i}$ be the smallest subtree of $T$ that is rooted 
at the root of $T$ and for which $\ell(\Gamma_{1,i}) = \Delta_{1,i}$.  
\end{enumerate} 

We will show that the trees $T_i$, $0 \leq i \leq r$, and $\Gamma_{j,i}$, 
$j \in \{0,1\}$ and $1 \leq i \leq \lceil \log | \ell(T) | \rceil$, 
form a separating tree system of the tree $T$. Let $u$ and $v$ be two distinct vertices of $T$. We number the levels in $T$ from $0$ to $r$, where the root is at level zero. Let $i$ be the level of $u$ and $j$ be the level of $v$ in $T$. We consider two cases: 

\begin{itemize}
    \item $i \neq j$. We may assume that $i < j$. Then $u$ is in $T_i$ whereas $v$ is not in $T_i$.
    \item $i = j$. Let $S_u$ be the set of labels of the leaves in the subtree of $u$, and similarly define $S_v$ for the vertex $v$. Note that $S_u$ and $S_v$ are two disjoint intervals of consecutive integers. By \cref{bit-or}, $b(S_u) \neq b(S_v)$ or $b(c(S_u)) \neq b(c(S_v))$. For any of these two inequalities, one of the bit positions in which they differ introduces the tree that separates $u$ from $v$.
\end{itemize}

The total number of elements used in the two sets of trees is at most $r + 2 \lceil \log{|\ell(T)|} \rceil + 1$, which is at most $r + 2 \lceil \log{n} \rceil + 1$.
\end{proof}

This result extends to all connected graphs. For a graph $G$ of radius $r$, we can consider a BFS tree rooted at a center of $G$ and use this spanning tree to find a separating system for $G$.

\section{Open Problems} \label{final}

The following computational complexity question is still open: Is there a polynomial time algorithm to compute $f(T)$, when $T$ is a tree? %The question of discovering a polynomial time algorithm for computing the value of $f(T)$ for a given tree $T$ remains an open problem.
However, we conjecture that determining the exact value of $f(G)$ for an arbitrary graph $G$ is NP-complete. Additionally, we conjecture the following:

\begin{conj}
    The problem of determining if $f(G) = \lceil \log{|V(G)|} \rceil$ for any given graph $G$ is NP-complete.
\end{conj}

% Another direction to take as future work is to study the relation of $f(G)$ to other graph parameters. For example, considering the results in \cref{poly-lowerbound}, we are curious to know if there is a relationship between the graph's connectivity or minimum degree and the parameter $f(G)$ for a given graph $G$. We list two questions regarding these parameters.

% \begin{prb}
%     What is the smallest integer $d = d(n)$ such that every connected $n$-vertex graph $G$ with minimum degree $d$ satisfies $f(G) = O(\log(n))$? It is not correct for $d = o(n / \log(n))$ --- a set of $n/(d+1)$ many $K_{d+1}$ graphs connected by cut edges together is a counter-example.
% \end{prb}

% \begin{prb}
%     What is the smallest integer $c = c(n)$ such that every $c$-connected $n$-vertex graph $G$ satisfies $f(G) = O(\log(n))$? Again, it is not correct for $c = o(n / \log(n))$ --- a path on $(n-c)$ vertices with $c$ apex vertices connected to all the $(n-c)$ vertices on the path is a counter-example.
% \end{prb}

% \cref{grid} implies the following simple corollary.
After studying the grid graphs in \cref{grid_sec}, \cref{grid} implies the following simple corollary.

\begin{cor} \label{prod}
Let $G_1$ and $G_2$ be two graphs. If $G_1$ and $G_2$ contain Hamiltonian paths, then $f(G_1 \square G_2) \le O(\log{|V(G_1)|} + \log{|V(G_2)|})$
\footnote{
For two graphs $G_1$ and $G_2$, the \defin{Cartesian graph product} of $G_1$ and $G_2$, denoted $G_1\square G_2$, is a graph whose vertex set is $V(G_1\square G_2):= V(G_1)\times V(G_2)$ and that contains an edge between distinct vertices $v=(v_1,v_2)$ and $w=(w_1,w_2)$ if and only if
\begin{inparaenum}[(i)]
    \item $v_1=w_1$ and $v_2w_2\in E(G_2)$; or
    \item $v_2=w_2$ and $v_1w_1\in E(G_1)$.
\end{inparaenum}
}.\end{cor}

In order to generalize this corollary, a natural question arises regarding the relationship of $f(G_1)$ and $f(G_2)$ for graphs $G_1$ and $G_2$, respectively, with their products (Cartesian product, strong product\footnote{For two graphs $G_1$ and $G_2$, the \defin{strong graph product} of $G_1$ and $G_2$, denoted $G_1\boxtimes G_2$, is a graph whose vertex set is $V(G_1\boxtimes G_2):= V(G_1)\times V(G_2)$ and that contains an edge between distinct vertices $v=(v_1,v_2)$ and $w=(w_1,w_2)$ if and only if
\begin{inparaenum}[(i)]
    \item $v_1=w_1$ and $v_2w_2\in E(G_2)$; or
    \item $v_2=w_2$ and $v_1w_1\in E(G_1)$; or
    \item $v_1w_1\in E(G_1)$ and $v_2w_2\in E(G_2)$.
\end{inparaenum}
}
, etc.). For instance, in line with \cref{prod} we can ask the same question for general graphs $G_1$ and $G_2$.% (we can extend this question to the strong product ($\boxtimes$)\footnote{For two graphs $G_1$ and $G_2$, the \defin{strong graph product} of $G_1$ and $G_2$, denoted $G_1\boxtimes G_2$, is a graph whose vertex set is $V(G_1\boxtimes G_2):= V(G_1)\times V(G_2)$ and that contains an edge between distinct vertices $v=(v_1,v_2)$ and $w=(w_1,w_2)$ if and only if
% \begin{inparaenum}[(i)]
%     \item $v_1=w_1$ and $v_2w_2\in E(G_2)$; or
%     \item $v_2=w_2$ and $v_1w_1\in E(G_1)$; or
%     \item $v_1w_1\in E(G_1)$ and $v_2w_2\in E(G_2)$.
% \end{inparaenum}
% } as well).

\begin{prb}
    Let $G_1$ and $G_2$ be connected graphs. What can we say about $f(G_1 \square G_2)$ as a function of $f(G_1)$ and $f(G_2)$?
\end{prb}

A question that could establish a connection between vertex-separating path systems and edge-separating path systems is as follows.

\begin{prb}
   Is there a relation between $f(G)$ and $f(\mathcal{L}(G))$, where $\mathcal{L}(G)$ is the line-graph of $G$?
\end{prb}

After exploring the maximal outerplanar graph in \cref{outer}, the next family of graphs to examine is maximal planar graphs. It is worth noting that some maximal planar graphs have a vertex-separating path system with linear size. However, determining the precise upper bound for this class of graphs is the next question we would like to ask.

\begin{prb}
    What is the value of $f(\Delta_n)$, where $\Delta_n$ is an $n$-vertex triangulation?
\end{prb}

As a final note, it would be worth exploring a tight upper bound for the size of vertex-separating tree systems on different graph classes, such as maximal outerplanar graphs, and triangulations.

\section*{Acknowledgement}

This research was initiated and conducted at the Tenth Annual Workshop on Geometry and Graphs, held at the Bellairs Research Institute in Barbados, February 3 – February 7, 2023. The authors are grateful to the organizers and to the participants of this workshop. We would like to thank the anonymous reviewers for careful reading of our manuscript. %The authors would like to thank Vida Dujmović for the evening discussions.

\nocite{*}
\bibliographystyle{abbrvnat}
% use the following instead if you encounter problems 
%\bibliographystyle{alpha}
\bibliography{main}
\label{sec:biblio}

\end{document}